\documentclass[conference]{IEEEtran}
\usepackage{mathrsfs}
\usepackage{amsfonts}
\usepackage{amssymb,bm}
\usepackage[cmex10]{amsmath}
\usepackage{multirow}
\usepackage{amsthm}
\usepackage{bbding}
\usepackage{color}
\usepackage[ruled,vlined,linesnumbered]{algorithm2e}
 \allowdisplaybreaks
\newtheorem{thm}{Theorem}
\newtheorem{defi}{Definition}
\newtheorem{lem}{Lemma}
\newtheorem{cor}{Corollary}
\newtheorem{prop}{Proposition}
\newtheorem{rmk}{Remark}
\newtheorem{conj}{Conjecture}
\newtheorem{const}{Construction}

\def\supp{\hbox{\rm{supp}}}
\def\Supp{{\rm{supp}}}
\renewcommand{\vec}[1]{\bm{#1}}

\def\md{{\;(\rm{mod}\;}}
\def\tr{{\rm{tr}}}

\def\wt{\hbox{\rm{wt}}}
\def\Tr{{\rm{Tr}}}
\def\ai{\hbox{\rm{AI}}}

\begin{document}

\title{Constructing Boolean Functions With Potential Optimal Algebraic Immunity Based on
Additive Decompositions of Finite Fields}

\author{\IEEEauthorblockN{Baofeng Wu}
\IEEEauthorblockA{State Key Lab of Information Security\\
Institute of Information Engineering\\
Chinese Academy of Sciences\\
Beijing 100093, China\\
Email: wubaofeng@iie.ac.cn} \and \IEEEauthorblockN{Qingfang Jin and
Zhuojun Liu}
\IEEEauthorblockA{Key Lab of Mathematics Mechanization\\
Academy of Mathematics and Systems Science\\
Chinese Academy of Sciences\\
Beijing 100190, China\\
Email: qfjin@amss.ac.cn\\and zliu@mmrc.iss.ac.cn} \and
\IEEEauthorblockN{Dongdai Lin}
\IEEEauthorblockA{State Key Lab of Information Security\\
Institute of Information Engineering\\
Chinese Academy of Sciences\\
Beijing 100093, China\\
Email: ddlin@iie.ac.cn}}

\maketitle

\begin{abstract}
We propose a general approach to construct cryptographic significant
Boolean functions of $(r+1)m$ variables based on the additive
decomposition $\mathbb{F}_{2^{rm}}\times\mathbb{F}_{2^m}$ of the
finite field $\mathbb{F}_{2^{(r+1)m}}$, where $r$ is odd and
$m\geq3$. A class of unbalanced functions are constructed first via
this approach, which coincides with a variant of the unbalanced
class of generalized Tu-Deng functions in the case $r=1$. This class
of functions have high algebraic degree, but their algebraic
immunity does not exceeds $m$, which is impossible to be optimal
when $r>1$. By modifying these unbalanced functions, we obtain a
class of balanced functions which have optimal algebraic degree and
high nonlinearity (shown by a lower bound we prove). These functions
have optimal algebraic immunity provided  a combinatorial conjecture
on binary strings which generalizes the Tu-Deng conjecture is true.
Computer investigations show that, at least for small values of
number of variables, functions from this class also behave well
against fast algebraic attacks.
\end{abstract}

\IEEEpeerreviewmaketitle

\section{Introduction}
Constructing Boolean functions satisfying all main criteria has
attracted a lot of attention in recent year. Among all these
criteria, optimal algebraic immunity seems necessary due to the
great success of algebraic attacks introduced (improved, more
definitely) by Courtious and Meier to some well-known
Boolean-function-based stream ciphers \cite{aa03}. Other criteria
for Boolean functions that can play as potential candidates in
designing such LFSR-based pseudo-random generators as filter
generators include balancedness, high algebraic degree and high
nonlinearity. Besides, because of the existence of the improved
algebraic attacks, the fast algebraic attacks (FAA's) \cite{faa03},
a good behavior against FAA's is also required for Boolean functions
to be usable in cryptography.

In fact, it is a big challenge to construct Boolean function with
optimal algebraic immunity together with  all other good
cryptographic properties and there has been little work on such a
topic until 2008. In their pioneering work \cite{CF08}, Carlet and
Feng constructed a classes of balanced functions with optimal
algebraic immunity, optimal algebraic degree, high nonlinearity and
good behavior against FAA's (verified by computers initially in
\cite{CF08} and confirmed by Liu et al. in \cite{liumc}
theoretically very recently). Their construction is based on finite
fields and the proof of optimal immunity of the constructed
functions is mostly based on univariate representations of Boolean
functions. Motivated by their idea of construction, Tu and Deng went
a further step. They constructed a class of balanced functions of
even number of variables with optimal algebraic degree, high
nonlinearity and potential optimal algebraic immunity. By
``potential" we mean that the optimal algebraic immunity is up to a
conjecture on binary strings (known as the Tu-Deng conjecture now)
which is not mathematically proved. In fact, their functions are
modified from functions belonging to a subclass of the well-known
$\mathcal{P}\mathcal{S}_{\text{ap}}$ class of bent functions. A
weakness of this class of functions is their immunity against FAA's
is bad \cite{carlet}. However, the idea of Tu and Deng's
construction is enlightening. Adopting similarly techniques, Tang et
al. constructed a class of functions satisfying all main criteria.
It is remarkable that the optimal algebraic immunity of this class
of functions is based on a combinatorial fact firstly conjectured by
Tang et al. and proved by Cohen and Flori \cite{cohen} afterwards.
Based on a general conjecture involving a parameter which can be
chosen rather freely  mentioned in \cite{tang} (known as the
generalized Tu-Deng conjecture), Jin et al. proposed a construction
of Boolean functions with optimal immunity covering those in
\cite{tu} and \cite{tang}. All the functions obtained in
\cite{tu,tang,jin} are constructed from a decomposition of the
finite field into a direct sum of a subfield and a copy of it, and
the proofs of (potential) optimal algebraic immunity of them are
mostly based on the so-called bivariate representations of Boolean
functions.

Note that the decompositions of finite fields used in
\cite{tu,tang,jin} are all additive ones. More precisely, the
additive group of a finite field  is decomposed into a direct sum
 of two additive groups with equal sizes to construct functions.
Therefore, to generalize the constructions in \cite{tu,tang,jin}, a
natural idea  is to use decompositions of additive groups of finite
 fields into  direct sums of additive groups with unequal sizes.
Besides, to study properties of functions constructed from such
kinds of decompositions,  the summands of a decomposition are
preferred both to be additive groups of certain finite fields.

In the present paper, we devote to realize this idea. By decomposing
the additive group of the finite field $\mathbb{F}_{2^{(r+1)m}}$
into a direct sum of additive groups of the finite fields
$\mathbb{F}_{2^{rm}}$ and $\mathbb{F}_{2^{m}}$ for an odd integer
$r\geq1$ and an integer $m\geq3$, we construct a class of
$(r+1)m$-variable unbalanced Boolean functions in a similar manner
with those in \cite{tu,tang,jin}. This class coincides with a
variant of the unbalanced class proposed in \cite{jin} when $r=1$,
but when $r>1$, some properties of functions belonging  to it are
different, say, their algebraic immunity will never be optimal.
However, after a modification of this class, we obtain a class of
balanced functions with optimal algebraic immunity provided a
combinatorial conjecture is true, but the proof of optimal algebraic
immunity of these functions in the case $r>1$ is quite different
from the proof in the case $r=1$, i.e. the proof of optimal
algebraic immunity of the balanced functions obtained in \cite{jin}.
In fact, in the case $r>1$, the first things that should be made
clear are,  how to represent functions defined from the additive
decomposition before-mentioned and how to study properties of such
functions under this kind of representation if we can find it.

The rest of the paper is organized as follows. In the following
section, we recall some basic notions about Boolean functions and
talk about bivariate representations of Boolean functions over
direct sums of finite fields. In Section III, we present a general
combinatorial  conjecture on binary strings. In Section IV, we
propose a class of unbalanced functions to make our idea of
constructing a class of balanced functions with good cryptographic
properties, which is proposed in Section V, more clear. Concluding
remarks are given in Section VI.

\section{Preliminaries}
In this section, we provide some basic notations and facts about
Boolean functions. For more details, we refer to \cite{carletbook}.

\subsection{Boolean functions and related basic notions}
Let $\mathbb{F}_{2}$  be the binary finite field and
$\mathbb{F}_{2}^{n}$ be the $n$-dimensional vector space over
$\mathbb{F}_{2}$. Any mapping from $\mathbb{F}_{2}^{n}$ to
$\mathbb{F}_{2}$ is called an  $n$-variable Boolean function.
Obviously, the set $\mathbb{B}_{n}$ consisting of all $n$-variable
Boolean functions forms an $\mathbb{F}_{2}$-algebra of dimension
$2^n$. For a Boolean function $f\in\mathbb{B}_n$, its support  is
defined as
$$\supp(f)=\{\vec x \in \mathbb{F}_{2}^{n}\mid f(\vec x)=1\},$$ and
the cardinality of this set, denoted by $\wt(f)$, is called its
Hamming weight. $f$ is called balanced if $\wt(f)=2^{n-1}$.
Furthermore, for another Boolean function $g\in\mathbb{B}_n$, the
distance between $f$ and $g$ is defined as ${\rm{d}}(f,g)=\wt(f+g)$.
 Abusing
notations, we also denote the Hamming weight of a vector $\vec
v\in\mathbb{F}_{2}^{n}$, i.e. the number of nonzero positions of
$\vec v$, to be $\wt(\vec v)$. Besides, for an integer $i$, we
denote by $\wt_n(i)$ the number of 1's in the binary expansion of
the reduction of $i$ modulo $(2^n-1)$ in the complete residue system
$\{0,1,\ldots,2^n-2\}$. Obviously, $\wt_n(-u)=n-\wt_n(u)$ when
$2^{n}-1\,\nmid\,u$.

By Lagrange interpolation, every $n$-variable Boolean function $f$
can be uniquely represented as
$$f(x_{1},\ldots,x_{n})=\sum_{I\subseteq \{1,2,\ldots,n\}}a_{I}
\,\prod_{i\in I}x_{i},~~a_I\in\mathbb{F}_{2}. $$ The deep reason for
the existence of such kinds of representations of Boolean functions,
often known as algebraic normal forms (ANF's) of Boolean functions,
lies in the isomorphism between $\mathbb{F}_{2}$-algebras
\[\mathbb{B}_n\cong\mathbb{F}_{2}[x_1,x_2,\ldots,x_n]/\langle x_1^2+x_1,\ldots,x_n^2+x_n\rangle.\]
Thanks to its ANF, we can define the algebraic degree of $f$, $\deg
f$,  to be the degree of $f(x_1,\ldots,x_n)$ as a multivariate
polynomial, i.e. $\deg f=\max_{I\subseteq \{1,2,\ldots,n\}}\{ |I|
\mid a_{I}\neq 0 \}$. Boolean functions of degree at most $1$ are
called affine functions. The minimum distance between $f$ and all
affine functions is called the nonlinearity of $f$ and denoted to be
$\mathcal{N}_f$. This notion characterizes how different is $f$ from
the simplest Boolean functions.

As is well known that the additive group of the finite field
$\mathbb{F}_{2^{n}}$ is an $n$-dimensional vector space over
$\mathbb{F}_{2}$, hence by Lagrange interpolation, the Boolean
function $f$ can also be represented by a univariate polynomial over
$\mathbb{F}_{2^{n}}$ of the form
$$f(x)=\sum_{i=0}^{2^{n}-1}f_{i}x^{i}.$$
However, since $f$ satisfies $f^2(x)=f(x)$ for any
$x\in\mathbb{F}_{2^{n}}$, there are some restrictions on the
coefficients of this kind of univariate representation. This kind of
representation implies that as $\mathbb{F}_{2}$-algebras,
$\mathbb{B}_{n}$ can be viewed as a subalgebra of
$\mathbb{F}_{2^n}/\langle x^{2^n}+x\rangle$. Comparing dimensions,
we can also obtain the isomorphism
\[\mathbb{F}_{2^n}/\langle x^{2^n}+x\rangle\cong\mathbb{B}_{n}\otimes_{\mathbb{F}_{2}}\mathbb{F}_{2^n}.\]
It can be deduced that, under its univariate representation, the
algebraic degree of $f$ is in fact
$$\deg f=\max_{0\leq i\leq 2^{n}-1}\{\wt_n(i)\mid f_{i} \neq 0\}.$$

\subsection{Bivariate representations of Boolean functions}

In fact,  representations of Boolean functions are more flexible
than what can be fully described. In this subsection, we introduce
the bivariate representations of Boolean functions, which have
already been mentioned in \cite{liu} without explaining details.

Assume $n=n_1+n_2$ for two integers $n_1,~n_2\geq1$. We can
decompose the additive group of $\mathbb{F}_{2^n}$ into a direct sum
of additive groups of $\mathbb{F}_{2^{n_1}}$ and
$\mathbb{F}_{2^{n_2}}$. Thus every $n$-variable Boolean function can
be viewed as a mapping from
$\mathbb{F}_{2^{n_1}}\times\mathbb{F}_{2^{n_2}}$ to
$\mathbb{F}_{2}$. By Lagrange interpolation, we can express
$f\in\mathbb{B}_{n}$ as
\begin{eqnarray*}
   f(x,y)&=&\sum_{(a,b)\in\mathbb{F}_{2^{n_1}}\times\mathbb{F}_{2^{n_2}}}
f(a,b)[1+(x+a)^{2^{n_1}-1}]  \\
   &&\qquad\quad\qquad\qquad\qquad\times[1+(y+b)^{2^{n_2}-1}]
\end{eqnarray*}
To expand this expression, we should do operations (multiplications
and additions) of elements from $\mathbb{F}_{2^{n_1}}$ and
$\mathbb{F}_{2^{n_2}}$. The smallest field in which these operations
can be done is the composite filed of $\mathbb{F}_{2^{n_1}}$ and
$\mathbb{F}_{2^{n_2}}$, i.e. $\mathbb{F}_{2^{[n_1,n_2]}}$, where
``$[\cdot,\cdot]$" represents the least common multiple of two
integers. Hence $f$ can actually be represented into the form
\begin{equation}\label{birep}
f(x,y)=\sum_{i=0}^{2^{n_1}-1}\sum_{j=0}^{2^{n_2}-1}f_{i,j}x^iy^j,~f_{i,j}\in\mathbb{F}_{2^{[n_1,n_2]}}.
\end{equation}
We call this kind of representation the bivariate representation of
$f$ over $\mathbb{F}_{2^{n_1}}\times\mathbb{F}_{2^{n_2}}$. It
follows that as $\mathbb{F}_{2}$-algebras, $\mathbb{B}_{n}$ can be
viewed as a subalgebra of $\mathbb{F}_{2^{[n_1,n_2]}}[x,y]/\langle
x^{2^{n_1}}+x,y^{2^{n_2}}+y\rangle$. Comparing dimensions we can
also deduce the isomorphism
\[\mathbb{F}_{2^{[n_1,n_2]}}[x,y]/\langle
x^{2^{n_1}}+x,y^{2^{n_2}}+y\rangle\cong\mathbb{B}_{n}\otimes_{\mathbb{F}_{2}}\mathbb{F}_{2^{[n_1,n_2]}}.\]
To obtain the ANF of $f$ from its bivariate representation, we just
need to choose two bases $\{\alpha_1,\ldots,\alpha_{n_1}\}$ and
$\{\beta_1,\ldots,\beta_{n_2}\}$ of $\mathbb{F}_{2^{n_1}}$ and
$\mathbb{F}_{2^{n_2}}$ over $\mathbb{F}_{2}$ respectively, and write
$x=\sum_{i=1}^{n_1}x_i\alpha_i$, $y=\sum_{j=1}^{n_2}y_j\beta_j$ for
two sets of variables $x_1,\ldots,x_{n_1}$ and $y_1,\ldots,y_{n_2}$
over $\mathbb{F}_{2}$, and then put them into $f(x,y)$. It can be
easily observed from this process that
$$\deg f\leq\max_{0\leq i\leq 2^{n_1}-1\atop 0\leq j\leq 2^{n_2}-1}\{\wt_{n_1}(i)+\wt_{n_2}(j)\mid f_{i,j}\neq 0\}.$$
The following lemma confirms that "$=$" actually holds.

\begin{prop}\label{birepdeg}
Assume $n=n_1+n_2$ and $f\in\mathbb{B}_{n}$ with the bivariate
representation \eqref{birep}.¡¡ Then
$$\deg f=\max_{0\leq i\leq 2^{n_1}-1\atop 0\leq j\leq 2^{n_2}-1}\{\wt_{n_1}(i)+\wt_{n_2}(j)\mid f_{i,j}\neq 0\}.$$
\end{prop}
\begin{IEEEproof}
Denote $\mathbb{F}_{2^{[n_1,n_2]}}[x,y]/\langle
x^{2^{n_1}}+x,y^{2^{n_2}}+y\rangle$ by $\mathcal{R}_n$ and let
$\mathbb{R}_n$ be the $\mathbb{F}_2$-subalgebra of $\mathcal{R}_n$
which is isomorphism to $\mathbb{B}_n$. For any $0\leq d\leq n$, let
${R}_d=\{h\in\mathbb{R}_n\mid
h=\sum_{i,j}h_{i,j}x^iy^j,~\wt_{n_1}(i)+\wt_{n_2}(j)\leq
d~\text{for~all}~i,j~\text{with}~h_{i,j}\neq0\}$ and
${B}_d=\{h\in\mathbb{B}_n\mid \deg h\leq d\}$, which are
$\mathbb{F}_2$-subspaces of $\mathbb{R}_n$ and $\mathbb{B}_n$
respectively. We just need to prove that
$\dim_{\mathbb{F}_2}R_d=\dim_{\mathbb{F}_2}B_d$ for all $0\leq d\leq
n$. First it is easy to see that
\[\dim_{\mathbb{F}_2}B_d=\sum_{k=0}^{d}{n\choose k}=\sum_{k=0}^{d}{n_1+n_2\choose k}.\]
To get $\dim_{\mathbb{F}_2}R_d$, we note that
$\bar{R}_d=R_d\otimes_{\mathbb{F}_2}\mathbb{F}_{2^{[n_1,n_2]}}$
where
\[\bar{R}_d:=
\left\{h\in\mathcal{R}_n\left|\begin{array} {c}
h=\sum_{i,j}h_{i,j}x^iy^j,\\
\wt_{n_1}(i)+\wt_{n_2}(j)\leq
d\\\text{for~all}~i,j~\text{with}~h_{i,j}\neq0
\end{array}\right.\right\}.\]
In fact, this can be observed from the isomorphism
$\mathbb{R}_n\otimes_{\mathbb{F}_2}\mathbb{F}_{2^{[n_1,n_2]}}=\mathcal{R}_n$
because essentially  the
``$\otimes_{\mathbb{F}_2}\mathbb{F}_{2^{[n_1,n_2]}}$" operation only
extends the definitional domain of coefficients of terms of
functions in $R_d$ to extend $R_d$ to be an
$\mathbb{F}_{2^{[n_1,n_2]}}$-vector space (more precisely, if $R_d$
is spanned by a basis $\{\beta_i\}$ over $\mathbb{F}_2$, then
$\bar{R}_d$ is spanned by the same basis over
$\mathbb{F}_{2^{[n_1,n_2]}}$), but all these terms $(x^iy^j)$'s and
the corresponding $(\wt_{n_1}(i)+\wt_{n_2}(j))$'s are not affected.
Therefore, we have
\[\dim_{\mathbb{F}_2}R_d=\dim_{\mathbb{F}_{2^{[n_1,n_2]}}}\bar{R}_d=
\sum_{0\leq k_1+k_2\leq d}{n_1\choose k_1}{n_2\choose k_2}.\] By the
Vandermonde's convolution for binomial coefficients \cite{graham},
we have
\[\sum_{0\leq k\leq d}{n_1+n_2\choose k}=\sum_{0\leq k_1+k_2\leq d}{n_1\choose k_1}{n_2\choose k_2}.\]
This completes the proof.
\end{IEEEproof}

\begin{rmk}
One may intuitively think the result of Proposition \ref{birepdeg}
natural. In fact, when $n_1=n_2=n/2$ for an even integer $n$, the
bivariate representations of Boolean functions  in this case were
frequently used in some authors' work (see e.g.
\cite{tu,tang,jin,liu}), and in all these work Proposition
\ref{birepdeg} was considered conventional and obvious, and was used
without given a proof of it. However, we can see from the proof of
Proposition \ref{birepdeg} that, even for the above simple case,
this result is far from obvious.
\end{rmk}

\subsection{Walsh transform of Boolean functions}
The Walsh transform of a Boolean function is a useful tool in
studying properties of it. The background of this concept is Fourier
analysis on finite Abelian groups. In nature, for a Boolean function
$f$, its Walsh transform is the Fourier transform of the complex
valued function $(-1)^f$ on a finite Abelian group. More precisely,
for $f\in\mathbb{B}_n$, its Walsh transform at any $\vec
a\in\mathbb{F}_{2}^{n}$ can be defined as
\[W_f(\vec a)=\sum_{\vec x\in\mathbb{F}_{2}^{n}}(-1)^{f(\vec x)+\vec a\cdot \vec x}
=\sum_{\vec x\in\mathbb{F}_{2}^{n}}(-1)^{f(\vec x)}\chi_{\vec
a}(\vec x),\] where ``$\cdot$" represents  the Euclidean inner
product of vectors and $\chi_{\vec a}$ is defined by $\chi_{\vec
a}(\vec x)=(-1)^{\vec a\cdot \vec x}$, $\forall \vec
x\in\mathbb{F}_{2}^{n}$. This is because the dual group
$\widehat{\mathbb{F}_{2}^{n}}$ of the additive Abeliean group
$\mathbb{F}_{2}^{n}$, i.e. the group formed by all additive
characters of $\mathbb{F}_{2}^{n}$,  is actually $\{\chi_{\vec
a}\mid \vec a\in\mathbb{F}_{2}^{n}\}$, all elements of which forms a
standard orthogonal basis of the space formed by all functions from
the group $\mathbb{F}_{2}^{n}$ to $\mathbb{C}^*$, the multiplication
group of the complex field. The Fourier transform of the complex
valued function $(-1)^f$ at $\vec\lambda\in\mathbb{F}_{2}^{n}$ is in
fact the coefficient before the term $\chi_{\vec\lambda}$ of the
Fourier expansion (i.e. the expansion under the basis $\{\chi_{\vec
a}\mid \vec a\in\mathbb{F}_{2}^{n}\}$) of $(-1)^f$. By this
definition, it can be easily derived that $f$ is balanced if and
only if $W_{f}(\vec0)=0$, and the nonlinearity of $f$ can be
equivalently expressed as
$$\mathcal{N}_{f}=2^{n-1}-\frac{1}{2}\max_{\vec a\in \mathbb{F}_{2}^{n}}|
W_{f}(\vec a)|.$$

According to the meaning of Walsh transform explained above, we are
clear that the Walsh transform of $f\in\mathbb{B}_n$ at any
$a\in\mathbb{F}_{2^n}$ can be defined as
$$W_{f}(a)=\sum_{x\in\mathbb{F}_{2^{n}}}(-1)^{f(x)+\tr_{1}^{n}(a
x)},$$ where $\tr_{1}^{n}(\cdot)$ is the trace function from
 $\mathbb{F}_{2^{n}}$ to $\mathbb{F}_{2}$, i.e.
 $\tr_{1}^{n}(x)=\sum_{i=0}^{n-1}x^{2^i}$ for any $x\in\mathbb{F}_{2^{n}}$.
 This is because in this case the dual group of $\mathbb{F}_{2^{n}}$
 is $\widehat{\mathbb{F}_{2^{n}}}=\{\chi_a\mid a\in\mathbb{F}_{2^{n}}\}$
 where for any $a\in\mathbb{F}_{2^{n}}$,
 $\chi_a(x):=(-1)^{\tr_{1}^{n}(ax)}$, $\forall x\in\mathbb{F}_{2^{n}}$.
 Furthermore, when $n=n_1+n_2$ and $f$ is viewed as a function from
 $\mathbb{F}_{2^{n_1}}\times\mathbb{F}_{2^{n_2}}$ to
 $\mathbb{F}_{2}$, the Walsh transform of $f$ at any $(a,b)\in\mathbb{F}_{2^{n_1}}\times\mathbb{F}_{2^{n_2}}$
 can be defined as
 \[W_f(a,b)=\sum_{(x,y)\in\mathbb{F}_{2^{n_1}}\times\mathbb{F}_{2^{n_2}}}
 (-1)^{f(x,y)+\tr_1^{n_1}(ax)+\tr_1^{n_2}(by)}.\]
This is because in this case
\[(\mathbb{F}_{2^{n_1}}\times\mathbb{F}_{2^{n_2}})^{\wedge}=
\widehat{\mathbb{F}_{2^{n_1}}}\times\widehat{\mathbb{F}_{2^{n_2}}}
=\{\chi_a\cdot\psi_b\mid
a\in\mathbb{F}_{2^{n_1}},b\in\mathbb{F}_{2^{n_2}}\},\] where for any
$a\in\mathbb{F}_{2^{n_1}}$, $b\in\mathbb{F}_{2^{n_2}}$,
$\chi_a(x):=(-1)^{\tr_1^{n_1}(ax)}$, $\psi_b(y):=\tr_1^{n_2}(by)$,
$\forall x\in\mathbb{F}_{2^{n_1}},~y\in\mathbb{F}_{2^{n_2}}$,
according to the following lemma (see e.g.
\cite[Exercise~5.4]{lidl}), the proof of which is simple and will be
omitted.
\begin{lem}
Let $G_1,~G_2$ be two Abelian groups. Then $\widehat{G_1\times
G_2}\cong\widehat{G_1}\times \widehat{G_2}$.
\end{lem}
Similarly, we also have such equivalent expression of the
nonlinearity of $f$ as
\[\mathcal{N}_{f}=2^{n-1}-\frac{1}{2}\max_{(a,b)\in\mathbb{F}_{2^{n_1}}\times\mathbb{F}_{2^{n_2}}}|
W_{f}(a,b)|.\]

\subsection{Algebraic immunity and immunity against FAA's of Boolean functions}

The notion of algebraic immunity of Boolean functions was introduced
in \cite{meier04} to measure the ability of LFSR-based pseudo-random
generators resisting algebraic attacks.
\begin{defi}
Let $f,g \in \mathbb{B}_{n}$. $g$ is called an annihilator of $f$ if
$fg=0$. The algebraic immunity  of $f$, $\ai(f)$, is defined to be
the smallest possible degree of the nonzero annihilators of $f$ or
$f+1$, i.e.
$$\ai(f)=\min_{0\neq g \in \mathbb{B}_{n}}\{\deg(g)\mid fg=0 \text{ or
}(f+1)g=0\}.$$
\end{defi}

It can be proved that the best possible value of the algebraic
immunity of $n$-variable Boolean functions is $\lceil{n}/{2}\rceil$
\cite{aa03}, thus functions attaining this upper bound are often
known as algebraic immunity optimal functions.

For a Boolean function $f\in\mathbb{B}_n$, optimal algebraic
immunity is necessary but not sufficient since when there exists a
function $g$ of low degree such that $gf$ is of a reasonable degree,
a fast  algebraic attack is feasible \cite{faa03}. In fact, $f$ is
 considered having best behavior against fast algebraic attacks if
 any pair of integers $(e,d)$ with $e<n/2$ and $e+d<n$ such that there  exists a nonzero function $g$
 of degree  $e$ satisfying that $gf$ is of degree  $d$, does not exist.

\section{Generalized Tu-Deng conjecture}

In \cite{tu} Tu and Deng proposed a combinatorial conjecture on
binary strings (known as the Tu-Deng conjecture now), based on which
they constructed a class of Boolean functions with optimal algebraic
immunity.

\begin{conj}[Tu-Deng]\label{TD}
Let $n=2k$ be an integer where $k\geq2$. For any $0\leq t \leq
2^k-2$, define
\[S_t=\left\{(a,b)\left|\begin{array}{c} 0\leq a,~ b\leq {2^k-2},\\
a+b\equiv t\md 2^k-1),\\
\wt_k(a)+\wt_k(b)\leq k-1
\end{array}\right.\right\}.\]
Then $|S_t|\leq 2^{k-1}$.
\end{conj}

As indicated in \cite[Remark~2]{tang}, this conjecture can be
generalized by replacing $a$ by $ua$ for any fixed integer $u$ with
$(u,2^k-1)$, and particularly, for the case $u=-2^l$ for some
integer $l\geq0$, a proof of this generalized conjecture can be
achieved \cite{cohen,jin}. Constructions of functions with optimal
algebraic immunity based on this generalized conjecture were also
obtained in \cite{jin}.

In the sequel we assume $n=(r+1)m$ for an odd integer $r\geq1$ and
an integer $m\geq3$, and pick an integer $u$ with $(u,2^m-1)$. We
propose a new combinatorial conjecture on binary strings which is a
more wide generalization of Conjecture \ref{TD}.

\begin{conj}\label{GTD}
For any $0\leq t \leq 2^m-2$, define
\[S_t=\left\{(a,b)\left|\begin{array}{c} 0\leq a\leq {2^{rm}-2},~0\leq b\leq {2^m-1},\\
ua+b\equiv t\md 2^m-1),\\
\wt_{rm}(a)+\wt_m(b)\leq n/2-1
\end{array}\right.\right\}.\]
Then $|S_t|\leq 2^{rm-1}$.
\end{conj}

\begin{rmk}\label{rkgtd}
It is easy to see that Conjecture \ref{GTD} generalizes the
conjecture proposed in \cite[Remark~2]{tang} (see also
\cite[Conjecture~3.3]{jin}) and of course, Conjecture \ref{TD}.
Indeed, the conjecture in \cite[Remark~2]{tang} can be viewed as the
$r=1$ case of Conjecture \ref{GTD} since in this case, the
cardinality of $S_t$ will not be affected if the restriction $0\leq
b\leq 2^m-1$ is replaced by $0 \leq b\leq 2^m-2$ for any $0\leq t
\leq 2^m-2$. Therefore, when $r=1$ and $u=-2^l$ for some integer
$l\geq0$, the conjecture is true according to \cite{cohen}.
\end{rmk}

We have checked the conjecture by computer experiments for (1)
$r=3$, $m=3,4,5,6,7$; (2) $r=5$, $m=3,4$; and (3) $r=7$, $m=3$, for
any $u$ with $(u,2^m-1)=1$, and for $r=3$, $m=8$ for $u=1$. Seeking
a proof of this conjecture, even the Tu-Deng conjecture which is a
very special case of it, is completely open. In addition, in the
case $r>1$ and $u=-2^l$ for some integer $l\geq0$, it seems
difficult to prove this conjecture though this can be done for
$r=1$.

\section{A class of unbalanced functions}

In the sequel, we fix a primitive element $\alpha$ of
$\mathbb{F}_{2^{rm}}$ and set $\beta=\alpha^{(2^{rm}-1)/(2^m-1)}$,
which is a primitive element of $\mathbb{F}_{2^{m}}$. For any
integer $0\leq s\leq 2^{rm}-2$, we denote $\Delta_s=\{\alpha^i\mid
s\leq i \leq s+2^{rm-1}-1\}$.

\begin{const}\label{const1}
Let $0\leq s\leq 2^{rm}-2$ be an integer. Define an $n$-variable
Boolean function
$f:\mathbb{F}_{2^{rm}}\times\mathbb{F}_{2^{m}}\rightarrow\mathbb{F}_{2}$
by setting
\[\supp(f)=\{(\gamma y^u,y)\mid y\in\mathbb{F}_{2^{m}}^*,~\gamma\in\Delta_s\}.\]
\end{const}

\begin{rmk}\label{rkbirep}
It is easy to see that the bivariate representation of $f$ over
$\mathbb{F}_{2^{rm}}\times\mathbb{F}_{2^{m}}$ can be written as
\[f(x,y)=g\left(\frac{x}{y^u}\right),\]
where $g$ is an $(rm)$-variable Boolean function with
$\supp(g)=\Delta_s$ (note that we always distinguish $x/0$ with $0$
in a finite field). We can see that this function can actually be
viewed as a $(2rm)$-variable generalized Tu-Deng function (i.e. a
function from \cite[Construction~4.1]{jin}) with the second
coordinate $y$ limited to the subfield $\mathbb{F}_{2^{m}}$ of
$\mathbb{F}_{2^{rm}}$. In particular, when $r=1$, it coincides with
the unbalanced generalized Tu-Deng function (see
\cite[Construction~4.1]{jin}).
\end{rmk}

  In the
following we discuss some  properties of the function defined in
Construction \ref{const1}.

\subsection{Bivariate representation and algebraic degree}

\begin{lem}\label{degn-1}
Let
$h:\mathbb{F}_{2^{rm}}\times\mathbb{F}_{2^m}\rightarrow\mathbb{F}_2$
be an $n$-variable Boolean function. Then $\deg h\leq n-2$ if and
only if $\wt(h)$ is even and
\[\sum_{(c_1,c_2)\in\Supp(h)}c_1=\sum_{(c_1,c_2)\in\Supp(h)}c_2=0.\]
\end{lem}
\begin{IEEEproof}
By Lagrange interpolation, the bivariate representation of $h$ over
$\mathbb{F}_{2^{rm}}\times\mathbb{F}_{2^m}$ can be written as
\begin{eqnarray*}
   &&h(x,y)\\
   &=& \sum_{(c_1,c_2)\in\Supp(h)}\left[1+(x+c_1)^{2^{rm}-1}\right]\left[1+(x+c_2)^{2^m-1}\right] \\
   &=&|\supp(h)|+\sum_{(c_1,c_2)\in\Supp(h)}(x+c_1)^{2^{rm}-1}\\&&+\sum_{(c_1,c_2)\in\Supp(h)}(x+c_2)^{2^{m}-1}\\
   &&+\sum_{(c_1,c_2)\in\Supp(h)}(x+c_1)^{2^{rm}-1}(x+c_2)^{2^{m}-1}.
\end{eqnarray*}
The coefficient of $x^{2^{rm}-1}y^{2^m-1}$, whose degree is $n$, is
$|\supp(h)|\mod2$; the coefficients of $x^{2^{rm}-2}y^{2^m-1}$ and
$x^{2^{rm}-1}y^{2^m-2}$, whose degrees are $n-1$, are
$\sum_{(c_1,c_2)\in\Supp(h)}c_1$ and
$\sum_{(c_1,c_2)\in\Supp(h)}c_2$ respectively. This completes the
proof.
\end{IEEEproof}

\begin{lem}\label{wt2m-1j}
Let $1\leq j\leq (2^{rm}-1)/(2^{m}-1)-1$ be an integer. Then
$\wt_{rm}((2^m-1)j)\leq rm-m$.
\end{lem}
\begin{IEEEproof}
For any $1\leq j\leq (2^{rm}-1)/(2^{m}-1)-1$, it is obvious that
\[\wt_{rm}\left((2^m-1)\left(\frac{2^{rm}-1}{2^{m}-1}-j\right)\right)=rm-\wt_{rm}((2^m-1)j).\]
Thus we need only to prove that $\wt_{rm}((2^m-1)j)\geq m$ for any
$1\leq j\leq (2^{rm}-1)/(2^{m}-1)-1$. Without loss of generality, we
can assume $j$ is odd. Denote by $\mathfrak{B}(a,b)$ the number of
borrows when calculating $a-b$ for two positive integers $a$ and $b$
with $a\geq b$. Then we have
\begin{eqnarray*}
  \wt_{rm}(2^mj-j) &=& \wt_{rm}(2^mj)-\wt_{rm}(j)+\mathfrak{B}(2^mj,j) \\
   &=&\mathfrak{B}(2^mj,j).
\end{eqnarray*}
It is easy to see that $\mathfrak{B}(2^mj,j)\geq m$ since $j$ is
odd.
\end{IEEEproof}

\begin{thm}\label{nbdeg}
Let $f$ be the Boolean function defined in Construction
\ref{const1}.  Then the bivariate representation of $f$ over
$\mathbb{F}_{2^{rm}}\times\mathbb{F}_{2^{m}}$ is
\begin{eqnarray*}
f(x,y)&=&\sum_{i=1\atop (2^m-1)\nmid
i}^{2^{rm}-2}\alpha^{-is}(1+\alpha^{-i})^{2^{rm-1}-1} x^iy^{2^m-1-\overline{ui}}\\
&&+\sum_{j=1}^{\frac{2^{rm}-1}{2^m-1}-1}
\alpha^{-(2^m-1)js}(1+\alpha^{-(2^m-1)j})^{2^{rm-1}-1}\\
&&\qquad\qquad\qquad\qquad\qquad\qquad\times\,
x^{(2^m-1)j}y^{2^m-1},
\end{eqnarray*}
where $\overline{ui}$ denotes the reduction of $ui$ modulo $(2^m-1)$
in the residue class $\{0,1,\ldots,2^m-2\}$ for any integer $1\leq
i\leq2^{rm}-2$. Therefore, $n-m\leq\deg f\leq n-2$.
\end{thm}
\begin{IEEEproof}
From the proof of \cite[Theorem~2]{tang} we know that the univariate
representation of the $(rm)$-variable function $g$ defined in Remark
\ref{rkbirep} is
\[g(x)=\sum_{i=1}^{2^{rm}-2}\alpha^{-is}(1+\alpha^{-i})^{2^{rm-1}-1}x^i.\]
Then the bivariate representation of $f$ follows from Remark
\ref{rkbirep}.

The algebraic degree of $f$ is $\max\{d_1,d_2\}$, where
\[d_1=\max\left\{\wt_{rm}(i)+\wt_{m}(2^m-1-\overline{ui})\left|\begin{array}{cc}
 1\leq i\leq 2^{rm}-2,\\(2^m-1)\nmid i\end{array}\right.\right\}\]
and
\[d_2=\max\left\{\wt_{rm}((2^m-1)j)+m\biggm| 1\leq j\leq \frac{2^{rm}-1}{2^m-1}-1\right\}.\]
By Lemma \ref{degn-1} we can get $d_1,d_2\leq n-2$. When
$i=2^{rm}-2$,
$\wt_{rm}(i)+\wt_{m}(2^m-1-\overline{ui})=rm-1+m-\wt_m(\overline{ui})=n-(\wt_m(\overline{-u})+1)$,
hence $n-m\leq d_1\leq n-2$. On the other hand, when
$j=(2^{rm}-1)/(2^m-1)-1$, $\wt_{rm}((2^m-1)j)=rm-m$, hence  we have
$d_2= rm=n-m$  from Lemma \ref{wt2m-1j}. Finally we get that
$n-m\leq\deg f\leq n-2$.
\end{IEEEproof}

\begin{rmk}
From the proof of Theorem \ref{nbdeg} we can see that:\\
 \noindent(1)
when $u=2^t$ for some non-negative integer $t$, $\deg
f=n-m$; and \\
(2) when $u=-2^t$ for some non-negative integer $t$, $\deg f=n-2$.
\end{rmk}

\begin{cor}\label{bentness}
Let $f$  be the Boolean function defined in Construction
\ref{const1}. Then $f$ is bent if and only if $r=1$ and $u=2^t$ for
some non-negative integer $t$.
\end{cor}
\begin{proof}
Since the algebraic degree of an $n$-variable bent function is at
most $n/2$ and $n/2\leq n-m\leq\deg f\leq n-2$ from Theorem
\ref{nbdeg}, we know that only when $r=1$, i.e. $n-m=n/2$, $f$ is
possibly  bent. Furthermore, when $u=2^t$ for some non-negative
integer $t$, it is clear that $f$ is bent (in fact, $f$ is
equivalent to a function belonging to the well-known
$\mathcal{P}\mathcal{S}_{\text{ap}}$ class of bent functions). To
prove this condition is also necessary, we should prove that $\deg
f=n/2=m$ implies $\wt_m(u)=1$. In fact, for any $1\leq i\leq 2^m-2$,
$\wt_{m}(i)+\wt_m(2^m-1-\overline{ui})=m+\wt_{m}(i)-\wt_m(\overline{ui})\leq
\deg f=m$, thus we have $\wt_m(i)\leq\wt_m(\overline{ui})$. Fixing
$i$ to be $2^m-2$, we get $m-1\leq\wt_m(\overline{-u})=m-\wt_m(u)$,
which implies that $\wt_m(u)\leq1$. Therefore, $\wt_m(u)=1$.
\end{proof}

\subsection{Algebraic immunity}

\begin{thm}\label{ainonban}
Let $f$ be the Boolean function defined in Construction
\ref{const1}. Then $\ai(f)\leq m$. In particular, $f$ has optimal
algebraic immunity provided that Conjecture \ref{GTD} is true when
$r=1$.
\end{thm}
\begin{IEEEproof}
Obviously, $1+y^{2^m-1}$ is an annihilator of $f$, whose degree is
$m$. This implies that $\ai(f)\leq m$.

When $r=1$, $f$ coincides with the function defined in
\cite[Construction~4.1]{jin} and Conjecture \ref{GTD} coincides with
\cite[Conjecture~3.3]{jin} according to Remark \ref{rkbirep} and
Remark \ref{rkgtd} respectively, so from \cite[Theorem~4.2]{jin} we
are clear that $\ai(f)=m=n/2$ if Conjecture \ref{GTD} is true.
\end{IEEEproof}

From Theorem \ref{ainonban} we can see that the algebraic immunity
of the functions from Construction \ref{const1} is not possible to
be optimal when $r>1$. However, it is interesting that they can be
modified to be functions with optimal algebraic immunity when
modified to be balanced functions. So in this case, our process to
obtain balanced functions with optimal algebraic immunity is
different from those in \cite{tu,tang,jin}, where balanced functions
with optimal algebraic immunity were all modified from unbalanced
ones with optimal algebraic immunity.

\section{A class of balanced functions with good cryptographic properties}

\begin{const}\label{const2}
Let $0\leq s,~l\leq 2^{rm}-2$ be two integers. Define an
$n$-variable Boolean function
$F:\mathbb{F}_{2^{rm}}\times\mathbb{F}_{2^{m}}\rightarrow\mathbb{F}_{2}$
 by setting
\begin{eqnarray*}
   \supp(F)&=& \{(\gamma y^u,y)\mid y\in\mathbb{F}_{2^{m}}^*,~\gamma\in\Delta_s\} \\
   &&\cup\{(\gamma,0)\mid \gamma\in\Delta_l\}.
\end{eqnarray*}
\end{const}

\begin{rmk}\label{rkbirepban}
It is easy to see that the bivariate representation of $F$ over
$\mathbb{F}_{2^{rm}}\times\mathbb{F}_{2^{m}}$ can be written as
\[F(x,y)=\left\{\begin{array}{ll}
g\left(\frac{x}{y^u}\right)&\text{if}~xy\neq0\\
\omega(x)&\text{if}~y=0,
\end{array}\right.\]
where $g$ and $\omega$ are  $(rm)$-variable functions with
$\supp(g)=\Delta_s$ and  $\supp(\omega)=\Delta_l$.
\end{rmk}

\begin{rmk}
It is easy to see that if $u_1$ and $u_2$ are chosen from the same
cyclotomic coset modulo $(2^m-1)$, then the functions defined from
$u_1$ and $u_2$ in Construction  \ref{const2} are linearly
equivalent.
\end{rmk}

Note that Construction \ref{const2} provides various ways to obtain
$n$-variable Boolean functions for an even integer $n$ since the
parameters, namely $m$,  $u$, $s$ and $l$, can be flexibly chosen.
In the following, we discuss some cryptographic properties of the
function $F$.

\subsection{Balancedness, bivariate representation and algebraic degree}
\begin{thm}
Let $F$ be the Boolean function defined in Construction
\ref{const2}. Then $F$ is balanced.
\end{thm}
\begin{IEEEproof}
It is obvious that $|\supp(f)|=(2^m-1)2^{rm-1}+2^{rm-1}=2^{n-1}$, so
$F$ is balanced.
\end{IEEEproof}

\begin{thm}
Let $F$ be the Boolean function defined in Construction
\ref{const2}. Then the bivariate representation of $F$ over
$\mathbb{F}_{2^{rm}}\times\mathbb{F}_{2^{m}}$ is
\begin{eqnarray*}
F(x,y)&=&\sum_{i=1\atop (2^m-1)\nmid
i}^{2^{rm}-2}\alpha^{-is}(1+\alpha^{-i})^{2^{rm-1}-1} x^iy^{2^m-1-\overline{ui}}\\
&&+\sum_{j=1}^{\frac{2^{rm}-1}{2^m-1}-1}
\alpha^{-(2^m-1)js}(1+\alpha^{-(2^m-1)j})^{2^{rm-1}-1}\\
&&\qquad\qquad\qquad\qquad\qquad\qquad\times\,
x^{(2^m-1)j}y^{2^m-1}\\
&&+\sum_{i=1}^{2^{rm}-2}\alpha^{-il}(1+\alpha^{-i})^{2^{rm-1}-1}
x^i(1+y^{2^m-1}).
\end{eqnarray*}
Therefore, $\deg F=n-1$, i.e. $F$ has optimal algebraic degree.
\end{thm}
\begin{IEEEproof}
It is easy to see from Remark \ref{rkbirepban} that the bivariate
representation of $F$ over
$\mathbb{F}_{2^{rm}}\times\mathbb{F}_{2^m}$ can be written as
\[F(x,y)=f(x,y)+\omega(x)(1+y^{2^m-1}),\]
where $f(x,y)$ is the function defined in Construction \ref{const1}.
 Then the representation of $F$ follows from Theorem \ref{nbdeg} and
 \cite[Theorem~2]{tang}.

Since $\omega(x)$ is in fact an $(rm)$-variable Carlet-Feng
function, we are clear that $\deg\omega=rm-1$ according to
\cite[Theorem~2]{CF08}, so the degree of $\omega(x)(1+y^{2^m-1})$ is
$rm-1+m=n-1$. However, by Theorem \ref{nbdeg} we have $\deg f\leq
n-2$. Finally we know that $\deg F=n-1$, which is optimal for a
balanced function.
\end{IEEEproof}

\subsection{Algebraic immunity}

In this subsection, we study the algebraic immunity of the functions
from Construction \ref{const2}. For the basic notions about BCH
codes and related results that will be used in the proof, we refer
to \cite{macwill}. Besides, the following lemma is also necessary.

\begin{lem}\label{veccon}
Let $\vec U,~\vec V\in\mathbb{F}_2^t$ be two binary vectors. Then
 $\wt(\vec U)+\wt(\vec V)\geq\wt(\vec U+\vec V)$.
\end{lem}
\begin{IEEEproof}
It is easy to see that  $\wt(\vec U+\vec V)=\wt(\vec U)+\wt(\vec
V)-\wt(\vec U\times\vec V)$, where $\vec U\times \vec V$ represents
the Hadamard product (i.e. bitwise multiplication) of $\vec U$ and
$\vec V$.
\end{IEEEproof}

\begin{thm}\label{aibalan}
Let $F$ be the Boolean function defined in Construction
\ref{const2}. Then $F$ has optimal algebraic immunity provided that
Conjecture \ref{GTD} is true.
\end{thm}
\begin{IEEEproof}
Since when $r=1$ the proof is almost the same with the proof of
\cite[Theorem~5.3]{jin}, we  need only  to deal with the case $r>1$.
We proceed by proving both $F$ and $F+1$ have no nonzero
annihilators of degree less than $n/2$ if Conjecture \ref{GTD} is
true.

Assume $h$ is an $n$-variable Boolean function with $\deg h<n/2$ and
$hF=0$. Write $h$ into its bivariate representation over
$\mathbb{F}_{2^{rm}}\times\mathbb{F}_{2^{m}}$ as
\[h(x,y)=\sum_{i=0}^{2^{rm}-1}\sum_{j=0}^{2^m-1}h_{i,j}x^iy^j.\]
From $\deg h<n/2< rm$ we know that $h_{i,j}=0$ for any $i,~j$ with
$\wt_{rm}(i)+\wt_m(j)\geq n/2$, which implies $h_{2^{rm}-1,j}=0$ for
any $0\leq j \leq 2^m-1$. Thus we can write $h$ into the form
\[h(x,y)=\sum_{i=0}^{2^{rm}-2}\sum_{j=0}^{2^m-2}h_{i,j}x^iy^j+
\sum_{i=0}^{2^{rm}-2}h_{i,2^m-1}x^iy^{2^m-1}.\] From
$h|_{\text{supp}(F)}=0$ we get that, for any
$y\in\mathbb{F}_{2^m}^*$, $\gamma\in\Delta_s$,
\begin{eqnarray*}
  h(\gamma y^u,y) &=&\sum_{i=0}^{2^{rm}-2}\sum_{j=0}^{2^m-2}h_{i,j}\gamma^iy^{ui+j}+\sum_{i=0}^{2^{rm}-2}h_{i,2^m-1}\gamma^i y^{ui}\\
   &=&\sum_{k=0}^{2^m-2}y^k\Bigg[\sum_{i=0}^{2^{rm}-2}h_{i,k-ui
   \,(\text{mod}\;2^m-1)}\gamma^i\\
   &&+\sum_{j=0}^{\frac{2^{rm}-1}{2^m-1}-1}
   h_{\tilde{u}k+j(2^m-1),2^m-1}\gamma^{\tilde{u}k+j(2^m-1)}\Bigg]\\
   &=&\sum_{k=0}^{2^m-2}h_k(\gamma)y^k\\&=&0,
\end{eqnarray*}
where $\tilde{u}$ is the integer satisfying $u\tilde{u}\equiv
1\md2^m-1)$ and $0\leq \tilde{u}k\leq 2^m-2$ is considered modulo
$(2^m-1)$, and
\begin{eqnarray*}
  h_k(\gamma) &=&\sum_{i=0}^{2^{rm}-2}h_{i,k-ui
   \,(\text{mod}\;2^m-1)}\gamma^i\\
   &&+\sum_{j=0}^{\frac{2^{rm}-1}{2^m-1}-1}
   h_{\tilde{u}k+j(2^m-1),2^m-1}\gamma^{\tilde{u}k+j(2^m-1)}.
\end{eqnarray*}
Therefore, for any $0\leq k\leq 2^m-2$, $h_k(\gamma)=0$ for any
$\gamma\in\Delta_s$. Viewing $h_k(\gamma)$ as a polynomial in
$\gamma$, we find that the vector of coefficients can be represented
as
\begin{eqnarray*}
\vec h_k&=&\big(h_{0,k},h_{1,k-u},\ldots,h_{\tilde{u}k,0},\ldots,h_{2^m-2,k+u},\\
  &&h_{2^m-1,k},h_{2^m,k-u},\ldots,h_{2^m-1+\tilde{u}k,0},\ldots,h_{2^{m+1}-3,k+u},\\
&&\ldots,h_{2^{rm}-2^m+\tilde{u}k,0},\ldots,h_{2^{rm}-2,k+u}\big)\\
&&+\big(0,\ldots,0,h_{\tilde{u}k,2^m-1},0,\ldots,0,h_{2^m-1+\tilde{u}k,2^m-1},0,\\
&&\ldots,0,h_{2^{rm}-2^m+\tilde{u}k,2^m-1},0,\ldots,0\big)\\
&:=&\vec h_k^{(1)}+\vec h_k^{(2)}.
\end{eqnarray*}
Now that $\vec h_k$ can be viewed as a codeword of certain BCH code
with designed distance $2^{rm-1}+1$, if it is not zero, the BCH
bound implies that $\wt(\vec h_k)\geq2^{rm-1}+1$. On the other hand,
Lemma \ref{veccon} and  Conjecture \ref{GTD} imply that
\[\wt(\vec h_k)=\wt\left(\vec h_k^{(1)}+\vec h_k^{(2)}\right)\leq
\wt\left(\vec h_k^{(1)}\right)+\wt\left(\vec h_k^{(2)}\right)\leq
2^{rm-1}.\] A contradiction follows and hence we have $\vec h_k=0$
for any $0\leq k\leq 2^m-2$, which leads to the fact that
$h_{i,0}=h_{i,2^m-1}$ for any $0\leq i\leq 2^{rm}-2$ with $i\equiv
\tilde{u}k\md 2^m-1)$, and $h_{i,k-ui}=0$ otherwise. Since  we have
the equality
\begin{eqnarray*}
   &&\bigcup_{0\leq k\leq 2^m-2}\{0\leq i\leq 2^{rm}-2\mid i\equiv
\tilde{u}k\md 2^m-1)\}  \\
   &&\qquad\qquad\quad=\{i\mid 0\leq i\leq 2^{rm}-2\},
\end{eqnarray*}
we are now clear that the annihilator  $h$ is of the form
\begin{eqnarray*}
   h(x,y)&=& \sum_{i=0}^{2^{rm}-2}(h_{i,0}x^i+h_{i,2^m-1}x^iy^{2^m-1}) \\
   &=&(1+y^{2^m-1})\sum_{i=0}^{2^{rm}-2}h_{i,0}x^i.
\end{eqnarray*}
In fact, the sums above are over all $i$'s with $\wt_{rm}(i)<n/2-m$.
Noting that $\{(\gamma,0)\mid \gamma\in\Delta_l\}\subseteq\supp(f)$,
we have $h(\gamma,0)=0$ for any $\gamma\in\Delta_l$, that is
\[\sum_{i=0}^{2^{rm}-2}h_{i,0}\gamma^i=0~\text{for~any}~\gamma\in\Delta_l.\]
Denote  $\vec h'=(h_{0,0},h_{1,0},\ldots,h_{2^{rm}-2,0})$. If $\vec
h'\neq\vec0$, the BCH bound implies that $\wt(\vec
h)\geq2^{rm-1}+1$; on the other hand, the restriction on the degree
of $h$ leads to
$$\wt(\vec
h')\leq\sum_{k=0}^{n/2-m-1}{rm\choose
k}<\sum_{k=0}^{\lfloor\frac{rm-1}{2}\rfloor}{rm\choose
k}\leq2^{rm-1}.$$ This contradiction implies that $\vec h'=\vec0$,
i.e. $h=0$.

As for $f+1$, the proof is almost the same. Assume $h$ is a
  Boolean function with $\deg h<n/2$
and $h(f+1)=0$ represented  as above. In this case, $h(\gamma
y^u,y)=0$ for any
$\gamma\in\mathbb{F}_{2^{rm}}^*\backslash\Delta_s$,
$y\in\mathbb{F}_{2^m}^*$, thus $\vec h_k$ can be viewed as a
codeword of certain BCH code with designed distance $2^{rm-1}$ and
the BCH bound implies that $\wt(\vec h_k)\geq2^{rm-1}$ if $\vec
h_k\neq\vec0$. On the other hand, $h(0,y)=0$ for any
$y\in\mathbb{F}_{2^m}$, which implies that $h_{0,k}=0$ for all
$0\leq k\leq 2^m-1$. Since $\wt_m(k)\leq m-1< n/2-1$ for any $0\leq
k\leq 2^m-2$, Lemma \ref{veccon} together with Conjecture \ref{GTD}
imply that $\wt(\vec h_k)\leq2^{rm-1}-1$, which lead to a
contradiction. Then we get that $h$ is of the form
\[h(x,y)=(1+y^{2^m-1})\sum_{i=0}^{2^{rm}-2}h_{i,0}x^i.\]
Further noting that $h(\gamma,0)=0$ for any
$\gamma\in\mathbb{F}_{2^{rm}}^*\backslash\Delta_l$, we get $\wt(\vec
h')\geq 2^{rm-1}$ by the BCH bound if $\vec h'\neq\vec0$, where
$\vec h'=(h_{0,0},h_{1,0},\ldots,h_{2^{rm}-2,0})$. However, from the
restriction on the degree of $h$, we have $\wt(\vec h')\leq
2^{rm-1}-1$. This contradiction leads to $h=0$.  We complete the
proof.
\end{IEEEproof}

\begin{rmk}
Set $D_{rm}=\sum_{k=0}^{n/2-m-1}{rm\choose k}$ and
$\Theta_t=\{\alpha^t,\ldots,\alpha^{t+D_{rm}-1}\}$ for any integer
$0\leq t\leq 2^{rm}-2$. Assume $l$ and $l'$ satisfy that
$\Theta_l\cap\Theta_{l'}=\emptyset$. Then from the proof of Theorem
\ref{aibalan}, it can be observed that if we set
$\supp(\omega)=\{(\gamma,0)\mid \gamma\in\Theta_l\cup C\}$ where $C$
is any subset of
$\mathbb{F}_{2^{rm}}^*\backslash(\Theta_l\cup\Theta_{l'})$
 with size $2^{rm-1}-D_{rm}$, the function $F$ constructed with this
 $\omega$ will also be balanced and have optimal algebraic immunity provided Conjecture
 \ref{GTD} is true. However, the algebraic degree of functions constructed in this
 manner might not be optimal.
\end{rmk}

\subsection{Nonlinearity}

Applying the classical technique of using Gauss sums to estimate
nonlinearity of Boolean functions constructed based on finite fields
(see, for example, \cite{CF08,tu,jin} and especially \cite{tang,liu}
), we can also obtain a lower bound of the nonlinearity of the
functions from Construction \ref{const2}. For simplicity, we use
"$\Tr$" and "$\tr$" to denote "$\tr^{rm}_1$" and "$\tr^m_1$"
respectively and denote $Q=2^{rm}$, $q=2^m$.

\begin{lem}[\cite{tang}]\label{lemtriang}
For every $0<x<\pi/2$,
\[\frac{1}{x}+\frac{x}{6}<\frac{1}{\sin x}<\frac{1}{x}+\frac{x}{4}.\]
\end{lem}

\begin{lem}\label{lemtrisum}
Let $T\geq2$ be an integer. Then
\[2T\left(\frac{\ln T}{\pi}+0.163\right)<\sum_{i=1}^{T-1}{\frac{1}{\sin\frac{\pi i}{2T}}}<
2T\left(\frac{\ln T}{\pi}+0.263\right)+\frac{3\pi}{8T}.\]
\end{lem}
\begin{IEEEproof}
From Lemma \ref{lemtriang} we have
\begin{eqnarray*}
   \sum_{i=1}^{T-1}{\frac{1}{\sin\frac{\pi i}{2T}}}&>&\frac{2T}{\pi}\sum_{i=1}^{T-1}\frac{1}{i}+\frac{\pi}{12T}\sum_{i=1}^{T-1}i  \\
   &\geq&\frac{2T}{\pi}\left(1+\sum_{i=2}^{T-1}\int_i^{i+1}\frac{\text{d} z}{z}\right)+\frac{\pi(T-1)}{24}\\
   &=&\frac{2T}{\pi}\left(1+\int_2^{T}\frac{\text{d} z}{z}\right)+\frac{\pi(T-1)}{24}\\
   &=&\frac{2T}{\pi}(\ln T+1-\ln2)+\frac{\pi(T-1)}{24}\\
   &=&2T\left(\frac{\ln T}{\pi}+\frac{1-\ln2}{\pi}+\frac{\pi}{48}\right)-\frac{\pi}{24}\\
   &>&2T\left(\frac{\ln T}{\pi}+0.163\right).
\end{eqnarray*}
On the other hand, we have
\begin{eqnarray*}
   \sum_{i=1}^{T-1}{\frac{1}{\sin\frac{\pi
   i}{2T}}}&<&\left(\frac{2T}{\pi}+\frac{\pi}{8T}+\frac{T}{\pi}+\frac{\pi}{4T}\right)+
   \frac{2T}{\pi}\sum_{i=3}^{T-1}{\frac{\frac{\pi}{2T}}{\sin\frac{\pi
   i}{2T}}}\\
   &<&\frac{3T}{\pi}+\frac{3\pi}{8T}+\frac{2T}{\pi}\sum_{i=3}^{T-1}
   \int_{\frac{\pi i}{2T}-\frac{\pi}{4T}}^{\frac{\pi i}{2T}+\frac{\pi}{4T}}
\frac{\text{d} z}{\sin z}\\
&<&\frac{3T}{\pi}+\frac{3\pi}{8T}+\frac{2T}{\pi}\int_{\frac{5\pi}{4T}}^{\frac{\pi}{2}}\frac{\text{d} z}{\sin z}\\
&=&\frac{3T}{\pi}+\frac{3\pi}{8T}-\frac{2T}{\pi}\ln\left(\tan\frac{5\pi}{8T}\right)\\
&\leq&\frac{3\pi}{8T}+\frac{2T}{\pi}\left(\ln T+1.5-\ln\frac{5\pi}{8T}\right)\\
&<&2T\left(\frac{\ln T}{\pi}+0.263\right)+\frac{3\pi}{8T}.
\end{eqnarray*}
\end{IEEEproof}

\begin{lem}\label{estmitation0}
Let $0\leq s\leq Q-2$ be an integer and
\[\Lambda_s=\sum_{\gamma \in \Delta_s}\sum_{y \in
\mathbb{F}_{q}^{*}}(-1)^{\Tr(\gamma y)} .\]Then
$|\Lambda_s|=2^{m-1}$ when $r=1$ and
\[|\Lambda_s|\leq \left[\frac{(n-2m)\ln2}{\pi}+0.263\right]2^{(n-m)/2}+2^{m-1}+1\]
when $r>1$.
\end{lem}
\begin{IEEEproof}
Let $\xi\in\mathbb{C}$ be a $(Q-1)$-th root of unity and
$\zeta=\xi^{N}$ where $N=(Q-1)/(q-1)$. Denote by $\chi_1$  the
primitive multiplication character of $\mathbb{F}_{Q}^*$  and define
the Gauss sums over $\mathbb{F}_{Q}$ as
\[G_1(\chi_1^\mu)=\sum_{x \in
\mathbb{F}_{Q}^{*}}\chi_1^{\mu}(x)(-1)^{\Tr(x)}\] for any $0\leq
\mu\leq Q-2$. It is well known that $G_1(\chi_1^0)=-1$ and
$|G_1(\chi^\mu)|=Q^{1/2}$ for any $1\leq \mu\leq Q-2$ \cite{lidl}.
By Fourier inversion we have, for any $0\leq i\leq Q-2$,
\[(-1)^{\Tr(\alpha^i)}=\frac{1}{Q-1}\sum_{\mu=0}^{Q-2}G_1(\chi_1^{\mu})\xi^{-\mu i}.\]
Hence we have
\begin{eqnarray*}
  \Lambda_s &=&\sum_{i=s}^{s+\frac{Q}{2}-1}\sum_{j=0}^{q-2}(-1)^{\Tr(\alpha^{i+Nj})}   \\
   &=&\frac{1}{Q-1}\sum_{i=s}^{s+\frac{Q}{2}-1}\sum_{j=0}^{q-2}
   \sum_{\mu=0}^{Q-2}G_1(\chi_1^{\mu})\xi^{-\mu (i+Nj)}\\
   &=&\frac{1}{Q-1}\sum_{\mu=0}^{Q-2}G_1(\chi_1^{\mu})
\sum_{i=s}^{s+\frac{Q}{2}-1}\xi^{-\mu i}\sum_{j=0}^{q-2}\zeta^{-\mu
j}.
\end{eqnarray*}
Note that
\[\sum_{i=s}^{s+\frac{Q}{2}-1}\xi^{-\mu
i}=\left\{\begin{array} {cl} \frac{Q}{2}&\text{if}~\mu=0\\
\xi^{-\mu
s}\frac{1-\xi^{-\mu\frac{Q}{2}}}{1-\xi^{-\mu}}&\text{otherwise},
\end{array}\right.
\]
\[\sum_{j=0}^{q-2}\zeta^{- \mu j}=\left\{\begin{array} {cl}
q-1&\text{if}~\mu\equiv0\md q-1)\\
0&\text{otherwise}.
\end{array}\right.
\]
Then we have
\begin{eqnarray*}
   \Lambda_s&=&\frac{1}{Q-1}\Bigg[-\frac{Q(q-1)}{2}\\&&
   + (q-1)\sum_{\mu=1\atop
(q-1)\mid\mu}^{Q-2}G_1(\chi_1^{\mu})\xi^{-\mu
s}\frac{1-\xi^{-\mu\frac{Q}{2}}}{1-\xi^{-\mu}}  \Bigg] .
\end{eqnarray*}
Note that when $r=1$, i.e. $Q=q$, the above formula yields
$\Lambda_s=-Q/2=-2^{m-1}$. When $r>1$, we can get that
\begin{eqnarray*}
   |\Lambda_s|&\leq&\frac{Q(q-1)}{2(Q-1)}+\frac{Q^{1/2}(q-1)}{Q-1}
\sum_{\mu=1\atop
(q-1)\mid\mu}^{Q-2}\left|\frac{1-\xi^{-\mu\frac{Q}{2}}}{1-\xi^{-\mu}}\right|\\
&=&\frac{Q(q-1)}{2(Q-1)}+\frac{Q^{1/2}(q-1)}{Q-1} \sum_{\mu=1\atop
(q-1)\mid\mu}^{Q-2}\left|\frac{1}{1+\xi^{-\mu/2}}\right|\\
&=&\frac{Q(q-1)}{2(Q-1)}+\frac{Q^{1/2}(q-1)}{Q-1}\sum_{k=1}^{N-1}\frac{1}{2\sin\frac{\pi
k}{2N}}.
\end{eqnarray*}
By Lemma \ref{lemtrisum} we have
\begin{eqnarray*}
&&|\Lambda_s|\\&\leq&\frac{Q(q-1)}{2(Q-1)}+\frac{Q^{1/2}(q-1)}{2(Q-1)}
\left[2N\left(\frac{\ln N}{\pi}+0.263\right)+\frac{3\pi}{8N}\right]
\\&<&   \frac{q}{2}+Q^{1/2}\left(\frac{1}{\pi}\ln \frac{Q}{q}+0.263\right)+\frac{3\pi
   Q^{1/2}(q-1)^2}{16(Q-1)^2}\\
&\leq&\left[\frac{(n-2m)\ln2}{\pi}+0.263\right]2^{(n-m)/2}+2^{m-1}+1.
\end{eqnarray*}
\end{IEEEproof}

\begin{lem}\label{estmitation1}
Let $0\leq s\leq Q-2$ be an integer. Denote
\[\Gamma_s=\sum_{\gamma \in \Delta_s}\sum_{y \in
\mathbb{F}_{q}^{*}}(-1)^{\Tr(\gamma y) + \tr(y^u)},
\]
where  $u$ is an integer with $(u,q-1)=1$. Then
\begin{eqnarray*}
  |\Gamma_s| &\leq&\left[\frac{(n-m)\ln2}{\pi}+0.263\right]2^{{n}/{2}}  \\
   &&-\left[\frac{(n-2m)\ln2}{\pi}+0.163\right]2^{{n}/{2}-m}+2.
\end{eqnarray*}
\end{lem}
\begin{IEEEproof}
Notations the same as those in the proof of Lemma \ref{estmitation0}
and further assume $\chi_2$ to be the primitive multiplication
character of  $\mathbb{F}_{q}^*$, and denote the Gauss sums over
 $\mathbb{F}_{q}$ by $G_2(\chi_2^\nu)$ for any $0\leq \nu\leq q-2$,
 i.e.
\[G_2(\chi_2^\nu)=\sum_{x \in
\mathbb{F}_{q}^{*}}\chi_2^{\nu}(x)(-1)^{\tr(x)}.\] We also have
$G_2(\chi_2^0)=-1$, $|G_2(\chi_2^\mu)|=q^{1/2}$ for any $1\leq
\nu\leq q-2$ and
\[(-1)^{\tr(\beta^j)}=\frac{1}{q-1}\sum_{\nu=0}^{q-2}G_2(\chi_2^{\nu})\zeta^{-\nu j}\]
for any $0\leq j\leq q-2$. Hence we have
\begin{eqnarray*}
  \Gamma_s &=& \sum_{i=s}^{s+\frac{Q}{2}-1}\sum_{j=0}^{q-2}(-1)^{\Tr(\alpha^i \beta^j) + \tr(\beta^{ju})} \\
   &=&\frac{1}{(Q-1)(q-1)}\sum_{i=s}^{s+\frac{Q}{2}-1}\sum_{j=0}^{q-2}
   \sum_{\mu=0}^{Q-2}G_1(\chi_1^{\mu})\xi^{-\mu (i+Nj)}\\&&
\qquad\qquad\qquad\qquad\qquad\times\sum_{\nu=0}^{q-2}G_2(\chi_2^{\nu})\zeta^{-\nu ju}\\
&=&\frac{1}{(Q-1)(q-1)}\sum_{\mu=0}^{Q-2}\sum_{\nu=0}^{q-2}G_1(\chi_1^{\mu})G_2(\chi_2^{\nu})\\
&&\qquad\qquad\qquad\times\sum_{i=s}^{s+\frac{Q}{2}-1}\xi^{-\mu
i}\sum_{j=0}^{q-2}\zeta^{- (\nu u+\mu)j}.
\end{eqnarray*}
Note that
\[\sum_{i=s}^{s+\frac{Q}{2}-1}\xi^{-\mu
i}=\left\{\begin{array} {cl} \frac{Q}{2}&\text{if}~\mu=0\\
\xi^{-\mu
s}\frac{1-\xi^{-\mu\frac{Q}{2}}}{1-\xi^{-\mu}}&\text{otherwise},
\end{array}\right.
\]
\[\sum_{j=0}^{q-2}\zeta^{- (\nu u+\mu)j}=\left\{\begin{array} {cl}
q-1&\text{if}~\nu u+\mu\equiv0\md q-1)\\
0&\text{otherwise}.
\end{array}\right.
\]
Since $\nu u+\mu\equiv0\mod(q-1)$ if and only if $\nu=0$ and
$\mu=k(q-1)$ for some $0\leq k\leq N-1$, or $\nu\equiv
q-1-\tilde{u}\mu\md q-1)$ and $(q-1)\nmid \mu$ where
$\tilde{u}u\equiv1\md q-1)$, we have
\begin{eqnarray*}
   &&\Gamma_s\\&=&\frac{1}{(Q-1)(q-1)}\Bigg[\frac{Q(q-1)}{2}\\&&+
(q-1)\sum_{\mu=1\atop
(q-1)\nmid\mu}^{Q-2}G_1(\chi_1^{\mu})G_2(\chi_2^{q-1-\tilde{u}\mu})\xi^{-\mu
s}\frac{1-\xi^{-\mu\frac{Q}{2}}}{1-\xi^{-\mu}}\\&&+
 (q-1)\sum_{\mu=1\atop
(q-1)\mid\mu}^{Q-2}G_1(\chi_1^{\mu})(-1)\xi^{-\mu
s}\frac{1-\xi^{-\mu\frac{Q}{2}}}{1-\xi^{-\mu}}  \Bigg] .
\end{eqnarray*}
Therefore, we can get that
\begin{eqnarray*}
   |\Gamma_s|&\leq&\frac{Q}{2(Q-1)}+\frac{Q^{1/2}q^{1/2}}{Q-1}
\sum_{\mu=1\atop
(q-1)\nmid\mu}^{Q-2}\left|\frac{1-\xi^{-\mu\frac{Q}{2}}}{1-\xi^{-\mu}}\right|\\&&+
\frac{Q^{1/2}}{Q-1}\sum_{\mu=1\atop
(q-1)\mid\mu}^{Q-2}\left|\frac{1-\xi^{-\mu\frac{Q}{2}}}{1-\xi^{-\mu}}\right|\\
&<&1+\frac{Q^{1/2}q^{1/2}}{Q-1}
\sum_{\mu=1}^{Q-2}\left|\frac{1}{1+\xi^{-\mu/2}}\right|\\
&&-\frac{Q^{1/2}(q^{1/2}-1)}{Q-1}\sum_{\mu=1\atop
(q-1)\mid\mu}^{Q-2}\left|\frac{1}{1+\xi^{-\mu/2}}\right|\\
&\leq& 1+\frac{Q^{1/2}q^{1/2}}{Q-1}
\sum_{\mu=1}^{Q-2}\frac{1}{2\sin\frac{\pi \mu}{2(Q-1)}}\\
&&-\frac{Q^{1/2}(q^{1/2}-1)}{Q-1}\sum_{k=1}^{N-1}\frac{1}{2\sin\frac{\pi
k}{2N}}.
\end{eqnarray*}
When $r=1$, i.e. $Q=q$ and $N=1$, by Lemma \ref{lemtrisum} we get
\begin{eqnarray*}
  |\Gamma_s| &\leq&1+\frac{q}{2(q-1)}\bigg[2(q-1)\left(\frac{\ln(q-1)}{\pi}+0.263\right)\\
  &&\qquad\qquad\qquad+\frac{3\pi}{8(q-1)}\bigg]  \\
   &\leq&2+\left(\frac{m\ln2}{\pi}+0.263\right)2^m.
\end{eqnarray*}
When $r>1$, by Lemma \ref{lemtrisum} we have
\begin{eqnarray*}
|\Gamma_s|&\leq&1+\frac{Q^{1/2}q^{1/2}}{2(Q-1)}\Bigg[
2(Q-1)\bigg(\frac{\ln(Q-1)}{\pi}+0.263\bigg)\\&&+\frac{3\pi}{8(
Q-1)}\Bigg] -\frac{Q^{1/2}(q^{1/2}-1)}{2(Q-1)}2N\left(\frac{\ln
N}{\pi}+0.163\right)\\
&<&2+\left(\frac{\ln Q}{\pi}+0.263\right)Q^{1/2}q^{1/2}\\
&&-\left(\frac{\ln N}{\pi}+0.163\right)\frac{Q^{1/2}}{q^{1/2}+1}\\
&\approx&\left[\frac{(n-m)\ln2}{\pi}+0.263\right]2^{{n}/{2}}\\
&&-\left[\frac{(n-2m)\ln2}{\pi}+0.163\right]2^{{n}/{2}-m}+2.
\end{eqnarray*}
Hence for any $r\geq 1$ approximately we can write  that
\begin{eqnarray*}
  |\Gamma_s| &\leq&\left[\frac{(n-m)\ln2}{\pi}+0.263\right]2^{{n}/{2}}  \\
   &&-\left[\frac{(n-2m)\ln2}{\pi}+0.163\right]2^{{n}/{2}-m}+2.
\end{eqnarray*}
\end{IEEEproof}

The following lemma is an equivalent formulation of
\cite[Theorem~5]{liumc}.
\begin{lem}\cite{liumc}\label{estmitation2}
Let $h$ be the Carlet-Feng function of $k$ variables. Then for any
$a\in\mathbb{F}_{2^k}$,
\[|W_h(a)|\leq\left(\frac{k\ln2}{\pi}+0.485\right)2^{k/2+1}.\]
\end{lem}

\begin{thm}\label{nlbal}
Let $F$ be the Boolean function defined in Construction
\ref{const2}. Then
\begin{eqnarray*}
  \mathcal{N}_F &\geq&2^{n-1}-\left[\frac{(n-m)\ln2}{\pi}+0.263\right]2^{{n}/{2}}\\
   &&-\left[\frac{(n-m)\ln2}{\pi}+0.485\right]2^{(n-m)/2}\\
   &&+\left[\frac{(n-2m)\ln2}{\pi}+0.163\right]2^{{n}/{2}-m}-2.
\end{eqnarray*}
\end{thm}
\begin{IEEEproof}
We compute $W_F(a,b)$ for any $(a,b)\in \mathbb{F}_{Q} \times
\mathbb{F}_{q}$. When $(a,b)=(0,0)$, we have $W_F(a,b)=0$ since $F$
is balanced. When $(a,b)\neq(0,0)$, we have
\begin{eqnarray*}
W_F(a,b)&=&-2\sum_{(x,y)\in \Supp(F)}(-1)^{\Tr(ax)+\tr(by)}\\
&=&-2\sum_{\gamma \in \Delta_s}\sum_{y \in
\mathbb{F}_{q}^{\ast}}(-1)^{\Tr(a\gamma y^u)+\tr(by)}\\&&-2\sum_{x
\in \Delta_l}(-1)^{\Tr(ax)}.
\end{eqnarray*}
If $a=0,~b\neq0$, then
\begin{eqnarray*}
W_F(a,b)&=&-2\sum_{\gamma \in \Delta_s}\sum_{y \in
\mathbb{F}_{q}^{\ast}}(-1)^{\tr(by)}-2\sum_{x
\in \Delta_l}1\\
&=&-2\times\frac{Q}{2}\times(-1) -2\times\frac{Q}{2}\\
&=&0.
\end{eqnarray*}
If $a\neq0,~b=0$, then
\begin{eqnarray*}
&&W_F(a,b)\\&=&-2\sum_{\gamma \in \Delta_s}\sum_{y \in
\mathbb{F}_{q}^{\ast}}(-1)^{\Tr(a\gamma y^u)}-2\sum_{x
\in \Delta_l}(-1)^{\Tr(ax)}\\
&=&-2\sum_{\gamma \in \Delta_{s'}}\sum_{y \in
\mathbb{F}_{q}^{\ast}}(-1)^{\Tr(\gamma y)}-2\sum_{x \in
\Delta_l}(-1)^{\Tr(ax)}\\&&\quad(\text{note~that}~\alpha^{s'}=a\alpha^s)\\
&=&-2\Lambda_{s'}+W_\omega(a),
\end{eqnarray*}
which leads to
\begin{eqnarray*}
   && |W_F(a,b)|\\ \\
   &\leq&\left\{\begin{array}{ll}
\left[\frac{(n-m)\ln2}{\pi}+0.485\right]2^{(n-m)/2+1}+2^m&\text{if}~r=1\\[10pt]
\begin{array}{l}\left[\frac{(2n-3m)\ln2}{\pi}+0.748\right]2^{(n-m)/2+1}\\+2^{m}+2\end{array}&\text{if}~r>1.
\end{array}\right.
\end{eqnarray*}
according to Lemma \ref{estmitation0} and Lemma \ref{estmitation2}.
If $ab\neq0$, it is easy to see that
\begin{eqnarray*}
W_F(a,b)&=&-2\sum_{\gamma \in \Delta_s}\sum_{y \in
\mathbb{F}_{q}^{\ast}}(-1)^{\Tr(b^{-u}a\gamma
y)+\tr(y^{\tilde{u}})}\\&&-2\sum_{x
\in \Delta_l}(-1)^{\Tr(ax)}\\
&=&-2\Gamma_{s'} +W_\omega(a)
\end{eqnarray*}
for some $0\leq s'\leq Q-2$, where $\tilde{u}u\equiv1\md q-1)$. Then
Lemma \ref{estmitation1} and Lemma \ref{estmitation2} implies that
\begin{eqnarray*}
   |W_F(a,b)|&\leq&2\left[\frac{(n-m)\ln2}{\pi}+0.263\right]2^{{n}/{2}}\\&&+2\left[\frac{(n-m)\ln2}{\pi}+0.485\right]2^{(n-m)/2}  \\
   &&-2\left[\frac{(n-2m)\ln2}{\pi}+0.163\right]2^{{n}/{2}-m}+4.
\end{eqnarray*}
Therefore, we finally get that
\begin{eqnarray*}
   &&\max_{(a,b)\in\mathbb{F}_Q\times\mathbb{F}_q}|W_F(a,b)|\\&=&\max\left\{\max_{a\in\mathbb{F}_Q^*}|W_F(a,0)|,~\max_{(a,b)\in\mathbb{F}_Q^*\times\mathbb{F}_q^*}|W_F(a,b)|\right\}\\
   &\leq&\left[\frac{(n-m)\ln2}{\pi}+0.263\right]2^{{n}/{2}+1}\\&&+\left[\frac{(n-m)\ln2}{\pi}+0.485\right]2^{(n-m)/2+1}  \\
   &&-\left[\frac{(n-2m)\ln2}{\pi}+0.163\right]2^{{n}/{2}-m+1}+4.
\end{eqnarray*}
Then we can complete the proof applying the relation
\[\mathcal{N}_F=2^{n-1}-\frac{1}{2}\max_{(a,b)\in\mathbb{F}_Q\times\mathbb{F}_q}|W_F(a,b)|.\]
\end{IEEEproof}

It can be seen from the expression of the lower bound of the
nonlinearity of $F$ given in Theorem \ref{nlbal} that, for a fixed
$n$, the bigger $m$ is, the higher the lower bound is. In
particular, when $m=n/2$, this lower bound is higher than the one
proposed in \cite{jin} and even higher that the one proposed in
\cite{tang} when $n\geq12$. See Table \ref{tablowbound} for the
comparison of lower bounds obtained in Theorem \ref{nlbal} and some
known ones for some values of $n$ in this case.

\begin{table*}[!htp]{\caption{Comparison of lower bounds of
nonlinearity in the case $n=2m$ \label{tablowbound}}}
\begin{center}\begin{tabular}{|c|c|c|c|c|c|c|c|c|c|c|c|}
  \hline
$n$& 6 & 8 & 10 & 12 & 14 & 16 & 18 & 20&22&24&26\\\hline LB in
Th.~\ref{nlbal}&20&102&457&1930&7936&32211&129863&521671&2091509&8376484&33528475\\\hline
LB in
\cite{jin}&18&93&429&1858&7762&31808&128949&519628&2086991&8366580&33506919\\\hline
LB in
\cite{tang}&20&102&458&1929&7931&32195&129823&521577&2091288&8376003&33527429\\
\hline
\end{tabular}\end{center}
\end{table*}

For small values of number of variables, we compute the exact values
of the nonlinearity of  $F$ for certain choices of $u$ (from
different cyclotomic cosets modulo ($2^m-1$)). Since the
computational results for the case $r=1$ have already presented  in
\cite{jin}, we need only to focus on the case $r>1$ here. Several
results for the case $r=3$ are listed in Table \ref{tabnlvalue}. By
comparing these values with nonlinearity of the Carlet-Feng
functions and the functions constructed in \cite{tang} in the
corresponding cases, it can be seen that, at least for these numbers
of variables, nonlinearity of functions from Construction
\ref{const2} is high.

\begin{table*}[!htp]{\caption{Nonlinearity of $F$ in the case $r=3$, $s=l=0$
\label{tabnlvalue}}}
\begin{center}\begin{tabular}{|c|c|c|c|c|c|}
  \hline
  $n$ & \multicolumn{2}{c|}{$\mathcal{N}_{F}$}  & $\mathcal{N}_{C\text{-}F}$ in \cite{CF08} & $\mathcal{N}_{T\text{-}C\text{-}T}$ in \cite{tang} &$2^{n-1}-2^{n/2-1}$   \\
  \hline
\multirow{2}{*}{12}   & $u=1$
&1982&\multirow{2}{*}{1970}&\multirow{2}{*}{1982}&\multirow{2}{*}{1984}\\\cline{2-3}
   & $u=6$     &1964&& &\\
   \hline
\multirow{2}{*}{16}   & $u=1$
&32408&\multirow{2}{*}{32530}&\multirow{2}{*}{32508}&\multirow{2}{*}{32512}\\\cline{2-3}
& $u=14$     &32406&&&\\
   \hline
\end{tabular}\end{center}
\end{table*}

\subsection{Immunity against FAA's}

As indicated in \cite{carlet},  when $r=1$ and $u=2^t$, the function
$F$ in Construction \ref{const2}, which can be viewed as a variant
of a balanced Tu-Deng function, behaves almost worst against FAA's.
The reason is that $F(x,y)$ only differs from $f(x,y)$, the function
defined in Construction \ref{const1}, when $y=0$, so for any linear
function $L(y)$ of $m$ variables, we have $L(y)F(x,y)=L(y)f(x,y)$,
which implies $\deg LF\leq m+1$ since in this case $\deg f=m=n/2$.
When $r>1$, a similar argument shows that, for any linear function
$L(y)$ of $m$ variables, $\deg LF\leq \deg f+1$. According to
Theorem \ref{nbdeg}, the degree of $f$ is $n-m$. Hence we are clear
that, for a fixed $n$, the smaller $r$ is (or the bigger $m$ is),
the worse behavior the functions from Construction \ref{const2}
against FAA's have, when $u=2^t$. For the case $u\neq 2^t$, the
behavior of functions from Construction \ref{const2} against FAA's
varies, and it is an interesting problem to study for what choice of
$u$ $F$ will play particularly well.

Fixing $s=l=0$ and choosing certain values of the parameters
$r,~m,~u$ (from different cyclotomic cosets modulo ($2^m-1$)), we do
some computer experiments to observe whether the pair $(e,d)$ with
$e<n/2$ and $e+d<n$ such that there is a function $h$ satisfying
$\deg h\leq e$ and $\deg hF\leq d$ exists. It turns
out that:\\
(1) in the cases $r=3$, $m=3$, (i.e. $n=12$), such pair with
$e+d\leq n-2$ does not exist for any possible $u$;\\
(2) in the case $r=3$, $m=4$, (i.e. $n=16$), such pair with $e+d\leq
n-2$ does not
exist  for any possible $u$;\\
(3) in the case $r=5$, $m=3$, (i.e. $n=18$), such pair with $e+d\leq
n-2$ does not exist, and the pairs $(3,14)$ and $(4,13)$ ($e+d=n-1$)
do not exist, for any possible $u$;\\
(4) in the case $r=3$, $m=5$, (i.e. $n=20$), such pair with $e+d\leq
n-2$ does not exist for any possible $u$ except $1$, and the pairs
$(1,15)$, $(2,15)$, $(3,15)$ and $(4,14)$ do not exist for $u=1$.
Besides, the pair $(4,15)$ ($e+d=n-1$) does not exist for $u=11$.
 These experimental results imply that the function
$F$ has good immunity against FAA's.

\section{Conclusion and further work}

We propose a general approach to construct Boolean functions with
good cryptographic properties based on decompositions of additive
groups of finite fields. A class of balanced functions with high
nonlinearity and optimal algebraic degree are constructed via this
approach. Algebraic immunity of these functions is optimal provided
a more generalized combinatorial conjecture on binary strings is
true, and immunity of them against fast algebraic attacks is also
good according to some computational results. This class of
functions covers some known classes of functions with (potential)
optimal algebraic immunity constructed based on additive
decompositions of finite fields.

Finally we should point out that, when $r=1$, behavior of the
function $F$ in Construction \ref{const2} against FAA's was
theoretically studied in \cite{liu}. Therefore, when $r>1$, how to
study behavior of $F$ against FAA's theoretically  will be a further
research topic of the authors.

\end{document}